\documentclass{article}
\usepackage{a4wide}
\usepackage[utf8]{inputenc}
\usepackage[english]{babel}
\usepackage{hyperref}
\usepackage{amsmath}
\usepackage{amsfonts}
\usepackage{amsthm}
\usepackage{graphicx}
\usepackage{enumitem}
\usepackage{verbatim}
\usepackage{xcolor}
\usepackage{tensor}
\theoremstyle{definition}
\newtheorem{defi}{Definition}[section]
\theoremstyle{plain}
\newtheorem{satz}[defi]{Theorem}
\theoremstyle{plain}
\newtheorem{propo}[defi]{Proposition}
\theoremstyle{plain}
\newtheorem{coro}[defi]{Corollary}
\theoremstyle{remark}
\newtheorem{bem}[defi]{Remark}

\begin{document}

\title{A Relation between Short-Term and Long-Term Arbitrage}
\author{P. Liebrich \\ \small Frankfurt Institute for Advanced Studies \\ \small Ruth-Moufang-Strasse 1, 60438 Frankfurt am Main, Germany}
\maketitle

\begin{abstract}
In this work a relation between a measure of short-term arbitrage in the market and the excess growth of portfolios as a notion of long-term arbitrage is established. The former originates from "Geometric Arbitrage Theory" \cite{GIGA} and the latter from "Stochastic Portfolio Theory" \cite{SPT}. Both aim to describe non-equilibrium effects in financial markets. Thereby, a connection between two different theoretical frameworks of arbitrage is drawn.
\end{abstract}

\section{Geometric Arbitrage Theory or short-term arbitrage}
The conventional mathematical theory of arbitrage (as e.~g. in \cite{DelbaenSchachermayer}) is actually a theory of asset pricing. The no-arbitrage property is postulated and then reformulated in terms of pricing rules for the market and single securities. The conventional way of study is to then introduce friction effects like transaction costs and liquidity constraints to adapt these rules to real markets.

There has then grown a literature on \textbf{gauge theoretic} and \textbf{differential geometric} models of \textbf{arbitrage} to do the opposite, allow arbitrage in the market formulation and find measures such that the classical no-arbitrage results are recovered in the respective limit (see e.~g. \cite{Ilinski, Young, SmithSpeed, GAT}). They are mostly ignored in classical finance since there is some arbitrariness in how to value arbitrage in a market and those models mostly approach it from the physical side. However, the article \cite{GIGA} built a general measure of arbitrage inspired by these considerations and derived an arbitrage strategy for frictionless markets based on that, thus giving justification for it. Take a filtered probability space $(\Omega,F,(F_t)_{t\in[0,\infty)},\mathbb{P})$ that satisfies the usual conditions. The model that underlies the securities' dynamics $S^\mu$ is given by \textbf{It\=o processes}. For convenience they shall be written in the form of \textbf{geometric Brownian motions} $$dS\indices{_t^\mu} = S\indices{_t^\mu} \left(\alpha\indices{_t^\mu} dt + \sigma\indices{_t^\mu_a} dW\indices{_t^a}\right), \quad \forall t \in [0,\infty), \mu=0,\ldots,n.$$ Note that the indices are spaced and they are implicitly summed over when one appears in an upper and one in a lower position, following physical notations. Both the drifts $\alpha^\mu$ as well as the covariances $\sigma\indices{^\mu_a}$ can depend on time $t$ and the state $\omega$ (adapted to the filtration) and can therefore also depend on the price itself to recover the above made claim of an It\=o process.

The whole market is assumed to consist of $m$ iid (standard) Brownian motions $W^a, a=1,\ldots,m$. The crucial assumption in this model is $m \leq n$ so that there is really an interdependence between the securities. This then leads to arbitrage opportunities when the corresponding basket of other securities yields a different price than the considered security itself. In practice, it may be one of the most challenging tasks to find the correct number of interdependences.

The geometric consideration comes into play when one observes that the stochastic differentials $dS\indices{_t^\mu}$ act like "cotangent vectors" in the "securities' portfolio manifold". As the stochastic differentials themselves are only notational objects, this has to be understood as a modeling idea. Nonetheless, the following quantities are defined in an exact manner. The idea is to decompose it into parts that evolve \emph{with the randomness} and parts \emph{orthogonal to it}. Therefore, they define the average and deviation quantities

\begin{defi}
\begin{align*}
\alpha_t &:= \frac{1}{n+1} \sum_{\mu=0}^n \alpha\indices{_t^\mu} \\
\sigma_{t,a} &:= \frac{1}{n+1} \sum_{\mu=0}^n \sigma\indices{_t^\mu_a} \\
\hat{\sigma}\indices{_t^\mu_a} &:= \sigma\indices{_t^\mu_a} - \sigma_{t,a}.
\end{align*}
\end{defi}

$\alpha_t$ can be regarded as the market return as well as $\sigma_{t,a}$ as a market volatility in the $a$-th Brownian motion. Regression against the specific randomness $\hat{\sigma}\indices{_t^\mu_a}$ yields

\begin{defi}\label{AlphaSplit}
$$\alpha\indices{_t^\mu} = \alpha_t + \beta\indices{_t^a} \hat{\sigma}\indices{_t^\mu_a} + \alpha_{t,A} J\indices{_t^{\mu A}}$$
for regression coefficients $\beta\indices{_t^a}, a=1,\ldots,m$ and orthonormal vectors $J\indices{_t^{\mu A}}, A=1,\ldots,k, \mu=0,\ldots,n$ with
\begin{align*}
J\indices{_{t,\mu}^A} J\indices{_t^{\mu B}} &= \delta^{AB} \\
\sum_{\mu=0}^n J\indices{_{t,\mu}^A} &= 0 \\
J\indices{_{t,\mu}^A} \hat{\sigma}\indices{_t^\mu_a} &= 0.
\end{align*}
\end{defi}

The dimension $k \leq n+1$ will in practice have to be guessed and measures the degree of dependence. It is a number for the dependent securities. So it can be assumed that often $k$ is very small. The last quantity $$\alpha\indices{_t^A} := J\indices{_{t,\mu}^A} \alpha\indices{_t^\mu}$$ is shown to be a measure of arbitrage. If all securities are independent, there will be nothing left; if there are many dependencies, there is a lot of potential arbitrage. The idea for the proof is to consider a change of numeraire by discounting with various portfolios. All other contributions in definition \ref{AlphaSplit} can be gauged away by a suited choice of numeraire and only the above term unavoidably remains. But they could also show a concrete trading strategy yielding the wealth $$V_T = \int_0^T \alpha_{t,A} \alpha\indices{_t^A} dt.$$ Their numerical studies showed that the arbitrage effects captured by this measure last for up to a minute and tend to vanish on larger timescales. It will be shown later that this contribution again shows up on macroscopic timescales in a portfolio context.

Another aspect of these theories is to rewrite arbitrage as the curvature of a suited gauge connection. Actually, this is the heart of those theories. A crucial requirement for trading strategies in arbitrage theory is the self-financing condition. This can geometrically be viewed as staying on the right path. As such, a covariant derivative seems appropriate, which keeps track of that path. The specific construct stems from \cite{Malaney, Weinstein}. They initially tackled the problem of comparing different economies' growths. A proper index should keep certain baskets fixed. This means that they can be seen as being parallel transported along a curve in the space spanned by prices and their weights. So the connection $A$ shall satisfy $$D_\gamma V = (d - A)|_\gamma V = 0$$ along a self-financing trajectory $\gamma$ with wealth $V$. The object of interest is originally defined as

\begin{defi}[Malaney-Weinstein connection]\label{MWconnection}
Let $(\phi_t)_{t\in[0,\infty)}$ be a predictable trading strategy and $V_t := \phi_{t,\mu} S\indices{_t^\mu}$ the wealth at each instant of time. The \textbf{Malaney-Weinstein connection} is defined as $$A_t := \frac{\phi_{t,\mu} dS\indices{_t^\mu}}{V_t}$$ and the \textbf{Malaney-Weinstein curvature} is $$R_t := dA_t$$ for all $t\in[0,\infty)$.
\end{defi}

Note that the connection is Abelian, so the definition of the curvature is really given by only the differential. The advance of \cite{GIGA} was to introduce a "stochastic" version of it. But it is still improperly defined, so formulate here the correct version. As noted in \cite{GAT}, the correct notion of derivative is given by Nelson's one \cite{Nelson}. But an analogous integral is needed as well:

\begin{defi}[Nelson derivative and normal expectation]\label{NelsonDE}
Let $\mathbb{P}^*$ be an equivalent probability measure on $(\Omega,F_t)$ for some $t\in\mathbb{R}_{\geq0}$, $(W_t^{*a})_{t\in[0,\infty)}$, $a=1,\ldots,m$ standard Brownian motions under $\mathbb{P}^*$ and let $(X_t)_{t\in[0,\infty)}$ be an adapted stochastic process. The \textbf{Nelson derivative} of $X$ at $t$ with respect to $\mathbb{P}^*$ is $$\mathcal{D}_{\mathbb{P}^*} X_t := \lim_{h\to0} \frac{\mathbb{E}_{\mathbb{P}^*}[X_{t+h} - X_t \mid F_t]}{h}$$ and the \textbf{normal expectation} of $X$ at $t$ with respect to $\mathbb{P}$ and $\mathbb{P}^*$ is $$\mathcal{E}_{\mathbb{P},\mathbb{P}^*}[X_t \mid F_t] := \mathbb{E}_{\mathbb{P}^*}[X_t \mid W\indices{_s^a} = W_s^{*a} = 0 \, \forall s \leq t, a=1,\ldots,m].$$
\end{defi}

Then the following statement holds:

\begin{satz}[Stochastic parallel transport of wealth]\label{Theorem3.1}
Define a time horizon $T \in \mathbb{R}_{>0}$, a shifted drift $\alpha_t^* := \alpha_t - \beta\indices{_t^a} \sigma_{t,a},$ shifted Brownian motions $W_t^{*a} := W\indices{_t^a} + \int_0^t \beta\indices{_s^a} ds$ and an equivalent martingale measure $\mathbb{P}^*$ defined by $$\frac{d\mathbb{P}}{d\mathbb{P}^*} := \exp\left(-\frac{1}{2} \int_t^T \beta_{s,a} \beta\indices{_s^a} ds + \int_t^T \beta_{s,a} dW_s^{*a}\right).$$ Moreover, let the Novikov condition $\mathbb{E}\left[e^{\frac{1}{2} \int_0^T \left|\beta_t\right|^2 dt}\right] < \infty$ hold true. If $(\phi_t)_{t\in[0,T]}$ is a self-financing strategy, then for all $t\in[0,T]$ the \textbf{stochastic Malaney-Weinstein connection} $$\Gamma_t := \frac{\phi_{t,\mu} \mathcal{D}_{\mathbb{P}^*} S\indices{_t^\mu}}{V_t}$$ satisfies $$\Gamma_t = \left(\frac{\sum_{\mu=0}^n (\sum_{A=1}^k) \alpha_{t,A} J\indices{_t^{\mu A}} \phi_{t,\mu} S\indices{_t^\mu}}{V_t} + \alpha_t^*\right) dt$$ and the value process obeys $$V_t = \mathcal{E}_{\mathbb{P},\mathbb{P}^*}\left[V_T e^{-\int_t^T \Gamma_s} \mid F_t\right].$$
\end{satz}
\begin{proof}
Although the notation is modified, the statement is given in \cite[Theorem 3.1]{GIGA}.
\end{proof}

\section{Stochastic Portfolio Theory or long-term arbitrage}
Another theory that handles markets with arbitrage, here on the classical stochastic level, is slowly entering mainstream mathematical finance. \textbf{Stochastic Portfolio Theory} \cite{SPT} seems to really be used in asset management as the founder of the theory has established successful investment management companies following that philosophy.

The \textbf{market} is again given by $n$ securities $S^i$, whereas a numeraire $S^0$ is not needed any more. (Roman indices generally begin at index 1 and Greek indices at 0.) Also, the assumption on the dynamics is geometric Brownian motions in the form $$dS\indices{_t^i} = S\indices{_t^i} \left(\alpha\indices{_t^i} dt + \sigma\indices{_t^i_a} dW\indices{_t^a}\right)$$ with (possibly) stochastic mean and covariance processes $(\alpha\indices{_t^i})_{t\in[0,\infty)}$ and $(\sigma\indices{_t^i_a})_{t\in[0,\infty)}$. So the theories are compatible with each other.

The parametrization used in the original treatise can be recovered by defining the \textbf{growth rate} $$\gamma\indices{_t^i} := \alpha\indices{_t^i} - \frac{1}{2} \sigma\indices{_t^i_a} \sigma\indices{_t^{ia}}.$$ It is often preferred to work with the logarithmic dynamics because it makes many things easier. In the whole theory, the time average is the more central object than the expectation value. Therefore, the law of iterated logarithm and other theorems come in handy.

The major difference to the short-term arbitrage model above is that the number $m$ of Brownian motions is assumed to equal the number of securities $m=n$ \cite{SPT} or is bigger $m \geq n$ \cite{SPTSurvey}. The reasoning is that in Stochastic Portfolio Theory there shall not exist trivial arbitrage. But in principle that would not necessarily have to lead to "easy" arbitrage and may even be realistic in real markets with their several dependencies. It is actually no essential assumption and most theorems remain valid for $m < n$.

Arbitrage in Stochastic Portfolio Theory comes into play in another way: It is theoretically allowed to build the portfolio \emph{after} knowing the prices. The rationale is that in long-term portfolios with little to no rebalancing this is not a big deal.

The object in the name of that theory is

\begin{defi}[Portfolio]\label{Portfolio}
Let $(\pi\indices{_t^i})_{t\in[0,\infty)}$ be bounded adapted processes for all $i=1,\ldots,n$ such that $$\sum_{i=1}^n \pi\indices{_t^i} = 1.$$ The process $\pi = (\pi^1, \ldots, \pi^n)$ is called a \textbf{portfolio}\index{portfolio!SPT}. The process $(Z\indices{_t^\pi})_{t\in[0,\infty)}$ that solves the SDE $$dZ\indices{_t^\pi} = \sum_{i=1}^n \pi\indices{_t^i} Z\indices{_t^\pi} \frac{dS\indices{_t^i}}{S\indices{_t^i}}$$ is called \textbf{portfolio value}\index{portfolio!value!SPT}.
\end{defi}

Stochastic Portfolio Theory deals with the long-term behavior of a market and portfolios. An object of interest is the \textbf{long-term growth rate} $\lim_{T\to\infty} \frac{1}{T} \log Z\indices{_T^\pi}$. And there are also different notions of arbitrage \cite{SPTSurvey}:

\begin{defi}[Relative arbitrage]
Let $\rho$ and $\pi$ be portfolios with $Z\indices{_0^\pi} = Z\indices{_0^\rho}$. $\pi$ is called a \textbf{relative arbitrage opportunity over a fixed horizon} $T \in \mathbb{R}_{>0}$ with respect to $\rho$ if there exists a number $q > 0$ such that $$\mathbb{P}(Z\indices{_T^\pi} \geq Z\indices{_T^\rho}) = 1, \quad \mathbb{P}(Z\indices{_T^\pi} > Z\indices{_T^\rho}) > 0,$$ and $$Z\indices{_t^\pi} \geq q Z\indices{_t^\rho} \quad \forall t\in[0,T] \text{ a.~s.}$$ $\pi$ is called a \textbf{superior long-term growth opportunity} with respect to $\rho$ if $$\underset{T\to\infty}{\lim\inf} \frac{1}{T} \log \frac{Z\indices{_T^\pi}}{Z\indices{_T^\rho}} > 0 \quad \text{a.~s.}$$
\end{defi}

Such arbitrage opportunities are mostly measured against the

\begin{defi}[Market portfolio]
The portfolio $(\mu\indices{_t^1},\ldots,\mu\indices{_t^n})_{t\in[0,\infty)}$ with $$\mu\indices{_t^i} := \frac{S\indices{_t^i}}{\sum_{j=1}^n S\indices{_t^j}}$$ is called the \textbf{market portfolio}. The market is called \textbf{coherent} if $$\lim_{T\to\infty} \frac{1}{T} \log \mu\indices{_T^i} = 0 \quad \text{a.~s. } \forall i=1,\ldots,n$$ holds. The market is called \textbf{non-degenerate} if there is an $\epsilon > 0$ such that $$x^i x^j \sigma_{t,ia} \sigma\indices{_{t,j}^a} \geq \epsilon x_i x^i \quad \forall t \in [0,\infty) \text{ a.~s.}$$
\end{defi}

It is an empirical fact that market-weighted portfolios usually underperform other weightings (see e.~g. \cite{RandomPaper}). This observation can also be shown within Stochastic Portfolio Theory:

\begin{satz}[Performance of the relative growth]
In a coherent and non-degenerate market, any portfolio $\pi$ that is constant over time with non-negative and at least two positive components, is a superior long-term growth opportunity with respect to the market portfolio $\mu$.
\end{satz}
\begin{proof}
See \cite[Proposition 2.1.9]{SPT}.
\end{proof}

There is then a whole playground to "generate" such portfolios. Part of the theory deals with finding a \textbf{generating function} $F$ on $\{x \in \mathbb{R}_{>0}^n: \sum_{i=1}^n x_i = 1\}$ for a portfolio $\pi$ such that $$\log \frac{Z\indices{_t^\pi}}{Z\indices{_t^\mu}} = \log F(\mu_t) + \Theta_t$$ for an adapted process $\Theta$ of bounded variation. There are many connections to information theory \cite[chapter 15]{InfoTheo} with the entropy being the toy model of a generating function.

\section{Geometric analysis of the portfolio return}
In this section Geometric Arbitrage Theory \cite{GIGA} shall be combined with Stochastic Portfolio Theory \cite{SPT}. Therefore, the number of Brownian motions is obviously set less than the number of securities $m < n$. In the spirit of Geometric Arbitrage Theory investigate the gauge behavior of generated portfolios and, more generally, of relative portfolios analogous to discounting by different securities:

\begin{propo}[Gauge transformation of relative portfolios]\label{APropo}
Let $\pi$ and $\rho$ be portfolios with $\inf_t |Z\indices{_t^\rho}| > 0$. The coefficients of the decomposition \ref{AlphaSplit} corresponding to the processes $$\hat{Z}\indices{_\rho^\pi} := \frac{Z^\pi}{Z^\rho}$$ equal $$\hat{\alpha}_t = 0, \quad \hat{\sigma}_{t,a} = 0, \quad \hat{\beta}\indices{_t^a} = \beta\indices{_t^a} - \rho_{t,j} \sigma\indices{_t^{j a}},$$ and those for the logarithmic processes $$\tilde{Z}\indices{_\rho^\pi} := \log \hat{Z}\indices{_\rho^\pi}$$ equal $$\tilde{\alpha}_t = 0, \quad \tilde{\sigma}_{t,a} = 0, \quad \tilde{\beta}\indices{_t^a} = \beta\indices{_t^a} - \frac{\pi_{t,j} + \rho_{t,j}}{2} \sigma\indices{_t^{j a}},$$ while $\alpha_{t,A}$ and $J\indices{_t^{iA}}$ remain unchanged.
\end{propo}
\begin{proof}
The ratio can be computed via the It\=o formula:
\begin{align*}
d\left(\hat{Z}\indices{_\rho^\pi}\right)_t &= \frac{dZ\indices{_t^\pi}}{Z\indices{_t^\rho}} -\frac{Z\indices{_t^\pi}}{(Z\indices{_t^\rho})^2} dZ\indices{_t^\rho} + 2 \cdot \frac{1}{2} \frac{Z\indices{_t^\pi}}{(Z\indices{_t^\rho})^3} (dZ\indices{_t^\rho})^2 - \frac{1}{(Z\indices{_t^\rho})^2} dZ\indices{_t^\pi} dZ\indices{_t^\rho} \\
&= \frac{Z\indices{_t^\pi}}{Z\indices{_t^\rho}} \pi_{t,i} \left(\alpha\indices{_t^i} dt + \sigma\indices{_t^i_a} dW\indices{_t^a}\right) - \frac{Z\indices{_t^\pi}}{(Z\indices{_t^\rho})^2} Z\indices{_t^\rho} \rho_{t,i} \left(\alpha\indices{_t^i} dt + \sigma\indices{_t^i_a} dW\indices{_t^a}\right) \\
&\quad + \frac{Z\indices{_t^\pi}}{(Z\indices{_t^\rho})^3} (Z\indices{_t^\rho})^2 \rho_{t,i} \rho_{t,j} \sigma\indices{_t^i_a} \sigma\indices{_t^{j a}} dt - \frac{1}{(Z\indices{_t^\rho})^2} Z\indices{_t^\pi} Z\indices{_t^\rho} \pi_{t,i} \rho_{t,j} \sigma\indices{_t^i_a} \sigma\indices{_t^{j a}} dt \\
&= \frac{Z\indices{_t^\pi}}{Z\indices{_t^\rho}} \left((\pi_{t,i} - \rho_{t,i}) \left(\alpha\indices{_t^i} dt + \sigma\indices{_t^i_a} dW\indices{_t^a}\right) + (\rho_{t,i} - \pi_{t,i}) \rho_{t,j} \sigma\indices{_t^i_a} \sigma\indices{_t^{j a}} dt\right) \\
&= \frac{Z\indices{_t^\pi}}{Z\indices{_t^\rho}} (\pi_{t,i} - \rho_{t,i}) \left((\alpha_t + \beta\indices{_t^a} \hat{\sigma}\indices{_t^i_a} + \alpha_{t,A} J\indices{_t^{iA}} - \rho_{t,j} \sigma\indices{_t^i_a} \sigma\indices{_t^{j a}}) dt + \sigma\indices{_t^i_a} dW\indices{_t^a}\right) \\
&= \frac{Z\indices{_t^\pi}}{Z\indices{_t^\rho}} (\pi_{t,i} - \rho_{t,i}) \left((0 + (\beta\indices{_t^a} - \rho_{t,j} \sigma\indices{_t^{j a}}) \hat{\sigma}\indices{_t^i_a} + \alpha_{t,A} J\indices{_t^{iA}} - 0) dt + \sigma\indices{_t^i_a} dW\indices{_t^a}\right).
\end{align*}
Note that in the last line the mean values $\alpha_t$ and $\sigma_{t,a}$ drop out because they have an independent factor $\sum_{i=1}^n (\pi_{t,i} - \rho_{t,i}) = 0$.

Furthermore, the logarithmic relative wealth fulfills
\begin{align*}
d\left(\log \hat{Z}\indices{_\rho^\pi}\right)_t &= \frac{1}{\hat{Z}\indices{_{t,\rho}^\pi}} d\hat{Z}\indices{_{t,\rho}^\pi} -\frac{1}{2(\hat{Z}\indices{_{t,\rho}^\pi})^2} (d\hat{Z}\indices{_{t,\rho}^\pi})^2 \\
&= (\pi_{t,i} - \rho_{t,i}) \left(((\beta\indices{_t^a} - \rho_{t,j} \sigma\indices{_t^{j a}}) \hat{\sigma}\indices{_t^i_a} + \alpha_{t,A} J\indices{_t^{iA}}) dt + \sigma\indices{_t^i_a} dW\indices{_t^a}\right) - \frac{1}{2} (\pi_{t,i} - \rho_{t,i}) (\pi_{t,j} - \rho_{t,j}) \sigma\indices{_t^i_a} \sigma\indices{_t^{j a}} dt \\
&= (\pi_{t,i} - \rho_{t,i}) \left(\left(\left(\beta\indices{_t^a} - \frac{\pi_{t,j} + \rho_{t,j}}{2} \sigma\indices{_t^{j a}}\right) \hat{\sigma}\indices{_t^i_a} + \alpha_{t,A} J\indices{_t^{iA}}\right) dt + \sigma\indices{_t^i_a} dW\indices{_t^a}\right).
\end{align*}
\end{proof}

This means that the relative wealths always already gauge away the mean drift $\alpha_t$ and volatility $\sigma_{t,a}$. This is quite clear because portfolios are only compared based on their specific characteristics and not on what the market does to every portfolio. The excess volatilities are scaled according to the portfolio differences $\pi_{t,i} - \rho_{t,i}$ and the sensitivities $\beta\indices{_t^a}$ to that excess volatility are reduced by the benchmark portfolio's weighted volatilities $\rho_{t,j} \sigma\indices{_t^{j a}}$. For the logarithmic relative wealth qualitatively the same happens. Still the part that cannot be gauged away, and is already in many cases the major drift part of the relative value, is the term $\alpha_{t,A} J\indices{_t^{iA}}$, which in the short-term analysis \cite{GIGA} is shown to be a driver of arbitrage. So in the very short term it is predominant because the other contributions are low on this scale and in the very long term they cancel out in the portfolio so that arbitrage exists again. One can make that more explicit by the following

\begin{coro}[Long-term growth rate of relative portfolios]
For portfolios $\pi$ and $\rho$ with $\inf_t |Z\indices{_t^\rho}| > 0$ the following applies: $$\lim_{T\to\infty} \frac{1}{T} \left(\log \frac{Z\indices{_T^\pi}}{Z\indices{_T^\rho}} - \int_0^T (\pi_{t,i} - \rho_{t,i}) \left(\left(\beta\indices{_t^a} - \frac{\pi_{t,j} + \rho_{t,j}}{2} \sigma\indices{_t^{j a}}\right) \hat{\sigma}\indices{_t^i_a} + \alpha_{t,A} J\indices{_t^{iA}}\right) dt\right) = 0 \quad \text{a.~s.}$$
\end{coro}
\begin{proof}
The terms are taken from proposition \ref{APropo}. The initial values obviously drop out in the long-term average. The Brownian part scales like $\sqrt{T}$ (compare \cite{IterLoga}) and therefore also drops out.
\end{proof}

If the volatilities are assumed to converge to each other in the long run, i.~e. $\hat{\sigma}\indices{_t^i_a} \to 0$, indeed only the geometric arbitrage part remains. Let the portfolio holdings and the geometric abitrage measures be approximately constant over time, this would mean for the long-term average of the relative growth rate $$\frac{1}{T} \log \frac{Z\indices{_T^\pi}}{Z\indices{_T^\rho}} \sim \frac{1}{T} \left((\pi_{T,i} - \rho_{T,i}) \alpha_{T,A} J\indices{_T^{iA}} T\right) = (\pi_{T,i} - \rho_{T,i}) \alpha_{T,A} J\indices{_T^{iA}}.$$ So the arbitrage scales with the difference in the portfolio holdings.

\begin{bem}
One can also consider the reverse situation. Assume that in the long run the geometric arbitrage term vanishes. Look at arbitrage similar as defined above within Stochastic Portfolio Theory. If the dynamics fulfills $$\beta\indices{_t^a} - \rho_{t,j} \sigma\indices{_t^{ja}} \begin{cases} > \sqrt{2t\log\log t} & \text{ if } \pi_{t,i} > \rho_{t,i}, \\ < -\sqrt{2t\log\log t} & \text{ if } \pi_{t,i} < \rho_{t,i}, \end{cases}$$ then one receives $$\underset{T\to\infty}{\lim\inf} \frac{1}{\sqrt{2T\log\log T}} \left(\frac{Z\indices{_T^\pi}}{Z\indices{_T^\rho}} - \frac{Z\indices{_0^\pi}}{Z\indices{_0^\rho}}\right) > 0 \quad \text{a.~s.}$$ and therefore $\frac{Z\indices{_T^\pi}}{Z\indices{_0^\pi}} > \frac{Z\indices{_T^\rho}}{Z\indices{_0^\rho}}$ with high probability over a sufficiently long period $T$. The numerical factor is chosen according to the law of iterated logarithm (see e.~g. \cite[Theorem 5.1]{IterLoga}). However, this is a quite strict range if one wanted to build a portfolio $\pi$ this way. It is very unlikely that the coefficients above fulfill this for larger times in either direction. This, in return, means that there will always be a probability for the portfolio $\pi$ to underperform $\rho$ over some horizon $\hat{T}$ even if there has been strong outperformance so far.
\end{bem}

\section{Malaney-Weinstein connection for portfolios}
Now have a closer look at the wealth SDE of Stochastic Portfolio Theory (definition \ref{Portfolio}). Its form immediately says that the portfolio value again follows a geometric Brownian motion. The analysis for the stock prices in \cite{GIGA} therefore carries over directly to the portfolio level. In particular, the portfolio value can be changed by the portfolio weights. As seen in proposition \ref{APropo}, one can also gauge away some coefficients, and there is the additional scaling freedom. This requires to consider \emph{relative wealths} and this is what is actually done in Stochastic Portfolio Theory \cite{SPT}. Then that gauge procedure is actually more general than the one for only the stocks since this can be recovered as the special case for $\pi\indices{_t^i} = \delta_k^i$ for some index $k$.

One might wonder how a covariant derivative and therefore the gauge connection might look like in analogy to the Malaney-Weinstein connection \ref{MWconnection}. The simple differential of the portfolio value is already given. So define the

\begin{defi}[Portfolio-valued Malaney-Weinstein connection]
Let $(\pi_t)_{t\in[0,\infty)}$ be a portfolio. The \textbf{portfolio-valued Malaney-Weinstein connection} is given by $$A\indices{_t^\pi} := \sum_{i=1}^n \pi\indices{_t^i} \frac{dS\indices{_t^i}}{S\indices{_t^i}}$$ and its curvature by $$R\indices{_t^\pi} := dA\indices{_t^\pi}$$ for all $t\in[0,\infty)$.
\end{defi}

This choice was taken so that the parallel transport equation $$DZ\indices{_t^\pi} := (d-A\indices{_t^\pi}) Z\indices{_t^\pi} = 0$$ holds. The difference to the original Malaney-Weinstein connection \ref{MWconnection} is that there the denominator was not given by the stock, but by the value process itself. On the one hand, this originates from still the sum and, on the other hand, from the absence of a numeraire here and the predictability of the trading strategy there. Explicitly, its components read $$A\indices{_{t,0}^\pi} := A\indices{_t^\pi}(\partial_{\pi_t}) = 0, \quad A\indices{_{t,1}^\pi} := A\indices{_t^\pi}(\partial_{S_t}) = \frac{\pi_t}{S_t}.$$ It transforms affinely under a change of numeraire resp. discounting at all:
\begin{align*}
\hat{A}\indices{_t^\pi} &= \sum_{i=1}^n \pi\indices{_t^i} \frac{d\hat{S}\indices{_t^i}}{\hat{S}\indices{_t^i}} \\
&= \sum_{i=1}^n \pi\indices{_t^i} \frac{d\left(\frac{S\indices{_t^i}}{S\indices{_t^0}}\right)}{\frac{S\indices{_t^i}}{S\indices{_t^0}}} \\
&= \sum_{i=1}^n \pi\indices{_t^i} \frac{dS\indices{_t^i}}{S\indices{_t^i}} - \sum_{i=1}^n \pi\indices{_t^i} \frac{dS\indices{_t^0}}{S\indices{_t^0}} \\
&=: A\indices{_t^\pi} + \Lambda_t.
\end{align*}
However, the affine term $\Lambda_t$ is not an exact form any more unless the portfolio is kept constant.

The curvature amounts to $$R\indices{_t^\pi} = \sum_{i=1}^n \frac{1}{S\indices{_t^i}} d\pi\indices{_t^i} \land dS\indices{_t^i}.$$ For comparison, its components read
\begin{align*}
R\indices{_{t,00}^\pi} &:= R\indices{_t^\pi}(\partial_{\pi_t}, \partial_{\pi_t}) = 0, &R\indices{_{t,01}^\pi} &:= R\indices{_t^\pi}(\partial_{\pi_t}, \partial_{S_t}) = \frac{1}{2S_t} \\
R\indices{_{t,10}^\pi} &:= R\indices{_t^\pi}(\partial_{S_t}, \partial_{\pi_t}) = -\frac{1}{2S_t}, &R\indices{_{t,11}^\pi} &:= R\indices{_t^\pi}(\partial_{S_t}, \partial_{S_t}) = 0.
\end{align*}
Under the stochastic integral such a differential form can be transformed into the space only spanned by time and the portfolios \cite{Zambrini}, so the expected curvature can be rewritten $$\mathbb{E}\left[R\indices{_t^\pi}\right] = \mathbb{E}\left[\sum_{i=1}^n \frac{\mathcal{D} S\indices{_t^i}}{S\indices{_t^i}} d\pi\indices{_t^i} \land dt\right] = \mathbb{E}\left[\sum_{i=1}^n \mathcal{D} \log S\indices{_t^i} d\pi\indices{_t^i} \land dt\right] =: \sum_{i=1}^n R\indices{_{t,i}^\pi} d\pi\indices{_t^i} \land dt.$$ There are two possibilities for this quantity to vanish: The portfolio weights $\pi\indices{_t^i}$ are kept constant in time or the logarithmic stock prices $\log S\indices{_t^i}$ are expected to stay constant. The latter case applies e.~g. for log-normal price processes with mean coefficient 0. In a perfect market this parameter is supposed to be given by the interest rate. In this model without numeraire it can very well be ignored. So $R\indices{_{t,i}^\pi} = 0$ constitutes an equilibrium condition. This is in line with other notions of market curvature in those gauge theoretic arbitrage models where it is always constructed as a (no-)arbitrage condition.

As a next step, turn the stochastic differentials intrinsically into ones in the Nelson sense. With the same assumptions as in theorem \ref{Theorem3.1}, the stochastic version of the above connection can be rewritten and has an even simpler expression than in theorem \ref{Theorem3.1}; the stock price does not appear any more at all. It now propagates the \emph{portfolio} instead of the \emph{wealth} process:

\begin{satz}[Portfolio-based stochastic parallel transport of wealth]\label{PortStockMW}
Let $(\pi_t)_{t\in[0,T]}$ be a portfolio process and let the other quantities from theorem \ref{Theorem3.1} be given. Then for all $t\in[0,T]$ the \textbf{stochastic portfolio-valued Malaney-Weinstein connection} $$\Gamma\indices{_t^\pi} := \sum_{i=1}^n \frac{\pi_{t,i} \mathcal{D}_{\mathbb{P}^*} S\indices{_t^i}}{S\indices{_t^i}}$$ satisfies $$\Gamma\indices{_t^\pi} = \left(\alpha_t^* + \sum_{i=1}^n \pi_{t,i} \alpha_{t,A} J\indices{_t^{i A}}\right) dt$$ and the portfolio process obeys $$Z\indices{_t^\pi} = \mathbb{E}_{\mathbb{P}^*}\left[Z\indices{_T^\pi} e^{-\int_t^T (\Gamma\indices{_s^\pi} + \sigma\indices{_t^i_a} \pi_{t,i} dW\indices{_s^a} + \beta_{s,a} dW_s^{*a})} \mid F_t\right] = \mathcal{E}_{\mathbb{P},\mathbb{P}^*}\left[Z\indices{_T^\pi} e^{-\int_t^T \Gamma\indices{_s^\pi}} \mid F_t\right].$$
\end{satz}
\begin{proof}
The correct Brownian motions under $\mathbb{P}^*$ are $W_t^{*a}$and the securities' dynamics in terms of $W_t^{*a}$ read $$dS\indices{_t^i} = S\indices{_t^i} \left(\left(\alpha_t^* + \alpha_{t,A} J\indices{_t^{i A}}\right) dt + \sigma\indices{_t^i_a} dW_t^{*a}\right).$$ Now the Nelson derivative works like "dividing the stochastic differentials: $\mathcal{D} \approx \frac{d}{dt}$" and the special definition here assures that one has to throw away Brownian motions with respect to $\mathbb{P}^*$. Therefore, the stochastic portfolio-valued Malaney-Weinstein connection amounts to
\begin{align*}
\Gamma\indices{_t^\pi} &:= \sum_{i=1}^n \frac{\pi_{t,i} \mathcal{D}_{\mathbb{P}^*} S\indices{_t^i}}{S\indices{_t^i}} \\
&= \sum_{i=1}^n \frac{\pi_{t,i} S\indices{_t^i}}{S\indices{_t^i}} \left(\alpha_t^* + \alpha_{t,A} J\indices{_t^{i A}}\right) dt \\
&= \left(\alpha_t^* + \sum_{i=1}^n \pi_{t,i} \alpha_{t,A} J\indices{_t^{i A}}\right) dt.
\end{align*}

Next, define the coefficients $$a_t := \sum_{i=1}^n \alpha\indices{_t^i} \pi_{t,i}, \quad b\indices{_t^a} := \sum_{i=1}^n \sigma\indices{_t^{ia}} \pi_{t,i},$$ which are nothing else than the portfolio-weighted mean return and mean volatility. Take another asset $\Lambda$ with dynamics $$d\Lambda_t = \Lambda_t \left(\left(-a_t + b\indices{_t^a} \beta_{t,a}\right) dt - \beta_{t,a} dW\indices{_t^a}\right).$$ The product process with the portfolio obeys
\begin{align*}
d(\Lambda Z^\pi)_t &= Z\indices{_t^\pi} d\Lambda_t + \Lambda_t dZ\indices{_t^\pi} + \frac{1}{2} d\Lambda_t dZ\indices{_t^\pi} \\
&= Z\indices{_t^\pi} \Lambda_t \left(\left(-a_t + b\indices{_t^a} \beta_{t,a}\right) dt - \beta_{t,a} dW\indices{_t^a}\right) + \Lambda_t Z\indices{_t^\pi} \left(a_t dt + b_{t,a} dW\indices{_t^a}\right) - \frac{1}{2} \Lambda_t Z\indices{_t^\pi} \beta_{t,a} b\indices{_t^a} dt \\
&= \Lambda_t Z\indices{_t^\pi} (b_{t,a} - \beta_{t,a}) dW\indices{_t^a}.
\end{align*}
This gives $$\log\left(\frac{\Lambda_T Z\indices{_T^\pi}}{\Lambda_t Z\indices{_t^\pi}}\right) = \int_t^T (b_{s,a} - \beta_{s,a}) dW\indices{_s^a}$$ or $$Z\indices{_t^\pi} = Z\indices{_T^\pi} \exp\left(\log\left(\frac{\Lambda_T}{\Lambda_t}\right) - \int_t^T (b_{s,a} - \beta_{s,a}) dW\indices{_s^a}\right).$$ Of course, this needs to remain valid if the conditional expectation with respect to $F_t$ is taken. But again in contrast to the corresponding statement in the original treatise \cite{GIGA}, this can meaningfully first be done in the last equation and then the second part does \emph{not} drop out. As before, this becomes correct when the modified expectation operator $\mathcal{E}$ introduced in definition \ref{NelsonDE} is used. Now calculate the logarithm of the $\Lambda$-process:
\begin{align*}
d\log(\Lambda)_t &= \frac{1}{\Lambda_t} \Lambda_t \left(\left(-a_t + b\indices{_t^a} \beta_{t,a}\right) dt - \frac{1}{2} \beta_{t,a} \beta\indices{_t^a} dt - \beta_{t,a} dW\indices{_t^a}\right) \\
&= \sum_{i=1}^n \left(-\alpha\indices{_t^i} \pi_{t,i} + \sigma\indices{_t^i_a} \pi_{t,i} \beta_{t,a}\right) dt - \frac{1}{2} \beta_{t,a} \beta\indices{_t^a} dt - \beta_{t,a} dW\indices{_t^a} \\
&= \sum_{i=1}^n \left((-\alpha_t^* - \alpha_{t,A} J\indices{_t^{i A}}) \pi_{t,i} - \frac{1}{2} \beta_{t,a} \beta\indices{_t^a}\right) dt - \beta_{t,a} dW\indices{_t^a} \\
&= -\Gamma\indices{_t^\pi} - \frac{1}{2} \beta_{t,a} \beta\indices{_t^a} dt - \beta_{t,a} dW\indices{_t^a}.
\end{align*}
Here the portfolio-valued Malaney-Weinstein connection appears again. Moreover under the Novikov condition, the Radon-Nikodým derivative $$\frac{d\mathbb{P}}{d\mathbb{P}^*} := \exp\left(-\frac{1}{2} \int_t^T \beta_{s,a} \beta\indices{_t^a} ds + \int_t^T \beta_{s,a} dW_s^{*a}\right)$$ is a martingale. Then the above expression for the advanced portfolio can be made more explicit by taking the conditional expectation value and using those statements:
\begin{align*}
Z\indices{_t^\pi} &= \mathbb{E}_\mathbb{P}[Z\indices{_t^\pi} \mid F_t] \\
&= \mathbb{E}_\mathbb{P}\left.\left[Z\indices{_T^\pi} \exp\left(\int_t^T d\log(\Lambda)_s - \int_t^T (b_{s,a} - \beta_{s,a}) dW\indices{_s^a}\right) \right\vert F_t\right] \\
&= \mathbb{E}_\mathbb{P}\left.\left[Z\indices{_T^\pi} \exp\left(\int_t^T \left(-\Gamma\indices{_s^\pi} - \frac{1}{2} \beta_{s,a} \beta\indices{_s^a} ds - \beta_{s,a} dW\indices{_s^a}\right) - \int_t^T (b_{s,a} - \beta_{s,a}) dW\indices{_s^a}\right) \right\vert F_t\right] \\
&= \mathbb{E}_{\mathbb{P}^*}\left.\left[\frac{d\mathbb{P}}{d\mathbb{P}^*} Z\indices{_T^\pi} \exp\left(-\int_t^T \left(\Gamma\indices{_s^\pi} + \frac{1}{2} \beta_{s,a} \beta\indices{_s^a} ds + b_{s,a} dW\indices{_s^a}\right)\right) \right\vert F_t\right] \\
&= \mathbb{E}_{\mathbb{P}^*}\left.\left[Z\indices{_T^\pi} \exp\left(-\int_t^T \left(\Gamma\indices{_s^\pi} + b_{s,a} dW\indices{_s^a} + \beta_{s,a} dW_s^{*a}\right)\right) \right\vert F_t\right].
\end{align*}
This expression still contains Brownian motions under $\mathbb{P}$ resp. $\mathbb{P}^*$, but no other perturbation terms. The normal expectation operator $\mathcal{E}$ cancels these contributions, the first one in the initial expectation under $\mathbb{P}$ and the second one in the expectation under $\mathbb{P}^*$.
\end{proof}

The representation of $\Gamma\indices{_t^\pi}$ in this theorem beautifully shows that its curvature vanishes if and only if the portfolio weights $\pi_{t,i}$ are kept constant or the $\alpha_{t,A}$ and hence arbitrage vanishes. The second statement is the stochastic equivalent to the required property $(d-A\indices{_t^\pi}) Z\indices{_t^\pi} = 0$ and as its solution $$\frac{Z\indices{_T^\pi}}{Z\indices{_t^\pi}} = e^{\int_t^T A\indices{_s^\pi}}$$ for a covariant derivative in the deterministic case.

\section{Conclusion}
In this article the two theories of Geometric Arbitrage Theory \cite{GIGA} and Stochastic Portfolio Theory \cite{SPT} have been brought together. The former has produced several gauge theoretic and differential geometric measures of arbitrage in the financial market. So far they have mostly been considered on the single asset side and as short-term deviations from the equilibrium. The latter deal with long-term effects on the portfolio side. So it is a rather natural question how to combine those two points of view.

It could be shown that the geometric decomposition of the securities' price dynamics (definition \ref{AlphaSplit}) applied to relative portfolios (proposition \ref{APropo}) explains the excess return arbitrage in a beautiful way. The contribution that is exogeneously given by the market is the same term as for single stocks. And this term was shown in \cite{GIGA} to generate arbitrage. It is therefore the same part that determines portfolio arbitrage in the long run.

Another contribution of this work is an accurate formalization of the Malaney-Weinstein connection \cite{Malaney, Weinstein} in the stochastic context. Moreover, it is also transferred to the portfolio context. Although this one is improper as a gauge connection, it could be shown that it correctly describes the parallel transport of portfolios (theorem \ref{PortStockMW}). Thereby, tools are established to evaluate whole portfolios in the long run with methods from gauge theory.

\end{document}